\documentclass[10pt]{article}
\usepackage[utf8]{inputenc}
\usepackage[UKenglish,USenglish,american]{babel}

\usepackage{palatino}
\usepackage{fullpage}

\usepackage[linesnumbered,ruled]{algorithm2e}

\usepackage{hyperref}
\usepackage[dvipsnames]{xcolor}
\usepackage{todonotes}
\usepackage{xspace}
\usepackage{complexity}
\usepackage{microtype}
\usepackage[linesnumbered,ruled]{algorithm2e}
\usepackage{tikz}
\usetikzlibrary{patterns}
\usetikzlibrary{decorations.pathmorphing}
\usepackage{amsmath, amsthm, amssymb}
\usepackage{enumerate}
\usepackage{thmtools}
\usepackage{doi}
\usepackage{adjustbox}
\usepackage{graphicx}
\usepackage{authblk}
\usepackage{subcaption}

\newtheorem{theorem}{Theorem}

\newtheorem{lemma}[theorem]{Lemma}

\newtheorem{conjecture}{Conjecture}
\newtheorem{observation}{Observation}

\newtheorem{corollary}[theorem]{Corollary}

\newtheorem{claim}{Claim}

\makeatletter
\@addtoreset{claim}{theorem}
\makeatother

\declaretheoremstyle[spaceabove=\topsep, spacebelow=\topsep, headfont=\itshape]{mystyle}

\newcommand{\source}{{\sf s}}
\newcommand{\target}{{\sf t}}

\definecolor{palette_bleuf}{HTML}{326C85}
\definecolor{palette_bleuc}{HTML}{679FB8}
\definecolor{palette_vert}{HTML}{399C54}
\definecolor{palette_jaune}{HTML}{FCD63E}
\definecolor{palette_beige}{HTML}{FAFEFF}
\definecolor{palette_rouge}{HTML}{9E0A1A}
\definecolor{palette_rose}{HTML}{FF9999}

\providecommand{\keywords}[1]{\textbf{Keywords:} #1}

\newenvironment{proofclaim}[1][]%
    {\noindent \emph{Proof.} {}{#1}{}}{\hfill
    $\Diamond$\vspace{1em}}
    
\title{Linear transformations between dominating sets in the TAR-model\thanks{This work was supported by ANR project GrR (ANR-18-CE40-0032).}}

\author[1]{Nicolas Bousquet}
\author[1]{Alice Joffard}
\author[2]{Paul Ouvrard}

\affil[1]{CNRS, LIRIS, Universit\'e de Lyon, Universit\'e Claude Bernard Lyon 1, Lyon, France
\thanks{firstname.lastname@liris.cnrs.fr}}
\affil[2]{Univ. Bordeaux, Bordeaux INP, CNRS, LaBRI, UMR5800, F-33400 Talence, France \thanks{paul.ouvrard@u-bordeaux.fr}}

\date{}
\begin{document}

\maketitle

\begin{abstract}
Given a graph $G$ and an integer $k$, a token addition and removal ({\sf TAR} for short) reconfiguration sequence between two dominating sets $D_\source$ and $D_\target$ of size at most $k$ is a sequence $S= \langle D_0 = D_\source, D_1 \ldots, D_\ell = D_\target \rangle$ of dominating sets of $G$ such that any two consecutive dominating sets differ by the addition or deletion of one vertex, and no dominating set has size bigger than $k$. 

We first improve a result of Haas and Seyffarth~\cite{Haas17}, by showing that if $k=\Gamma(G)+\alpha(G)-1$ (where $\Gamma(G)$ is the maximum size of a minimal dominating set and $\alpha(G)$ the maximum size of an independent set), then there exists a linear {\sf TAR} reconfiguration sequence between any pair of dominating sets. 

We then improve these results on several graph classes by showing that the same holds for $K_{\ell}$-minor free graph as long as $k \ge \Gamma(G)+O(\ell \sqrt{\log \ell})$ and for planar graphs whenever $k \ge \Gamma(G)+3$.
Finally, we show that if $k=\Gamma(G)+tw(G)+1$, then there also exists a linear transformation between any pair of dominating sets.


\end{abstract}

\section{Introduction}
\paragraph{General introduction.}
{\em Reconfiguration problems} model dynamic situations where we are given an instance $\mathcal{I}$ of a combinatorial search problem $\Pi$ and we want to find a step-by-step transformation between feasible solutions of $\mathcal{I}$ such that each intermediate solution satisfies the two following properties (i) it is also a feasible solution of $\mathcal{I}$; and (ii) it is obtained from the previous one by applying a specified (and unique) rule, called {\em reconfiguration rule}. Such a transformation between two solutions $S_\source$ and $S_\target$ of $\mathcal{I}$ is called a {\em reconfiguration sequence between $S_\source$ and $S_\target$}, and is denoted by $\langle S_0 = S_\source, S_1, S_2, \ldots, S_\ell = S_\target \rangle$. A reconfiguration sequence does not always exist and some solutions may even be frozen, meaning that they cannot be modified at all. Ito et al.~\cite{Ito11} initiated a systematic study of the complexity of reconfiguration problems. For a more complete overview of the field, the reader is referred to the surveys of Van den Heuvel~\cite{Heuvel13}, Nishimura~\cite{Nishimura18}, or Mynhardt and Nasserasr~\cite{MynhardtSurvey}.

\medskip
It is often interesting to study reconfiguration problems by looking at the {\em reconfiguration graph}. The vertices of the reconfiguration graph are the feasible solutions of the instance $\mathcal{I}$ of the problem $\Pi$, and two vertices (solutions of $\mathcal{I}$) are adjacent if and only if one solution can be obtained from the other by applying the specified reconfiguration rule. In this paper, we focus on the reconfiguration of dominating sets. A \emph{dominating set} is a subset $D$ of vertices such that each vertex is in $D$ or has at least one neighbor in $D$. One can represent a dominating set as a set of tokens, where exactly one token is placed on each vertex that is part of the dominating set. Then, one needs to define an operation that allows to transform a dominating set into another one. In the literature, three kinds of operations have mainly been studied: {\em Token Sliding} (at each step, one can slide exactly one token along an edge), {\em Token Jumping} (at each step, one can move exactly one token to any vertex which does not already contain a token), or {\em Token Addition and Removal} (at each step, one can add exactly one token or remove exactly one token).
One can observe that, for the first two rules, the size of each solution remains the same all along the transformation while it is modified at each step in the token addition and removal operation. 
In this paper, we only consider the token addition and removal rule, denoted by {\sf TAR} for short. 

\paragraph{Dominating set reconfiguration.}
One can indeed always transform a solution $S_\source$ into another one $S_\target$ if we do not bound the maximum size of the intermediate solutions: we first add one by one all the vertices in $S_\target \setminus S_\source$ to $S_\source$, and then remove each vertex in $S_\source \setminus S_\target$. If tokens are agents or equipment, there is not necessarily enough agents to perform this transformation. The problem becomes much harder when we have a threshold on the size of each solution we cannot exceed. 

Let $G=(V,E)$ be a graph, and $k$ be an integer. The $k$-reconfiguration graph (also known as {\em $k$-dominating graph}) is a graph $\mathcal{R}_k(G)$ whose vertices are the dominating sets of $G$ of size at most $k$, and two dominating sets $D_1$ and $D_2$ are adjacent if and only if the size of their symmetric difference $|D_1 \, \triangle \, D_2|$ is equal to one. In other words, $D_2$ can be obtained from $D_1$ by removing or adding exactly one token. Hence, there exists a reconfiguration sequence between two dominating sets  $D_\source$ and $D_\target$ both of size at most $k$ under the {\sf TAR} rule with threshold $k$ (denoted by {\sf TAR}($k$) rule for short) if and only if there is a path in $\mathcal{R}_k(G)$ between $D_\source$ and $D_\target$.  

\begin{figure}[bt]
    \centering
    \includegraphics[width=0.98\textwidth]{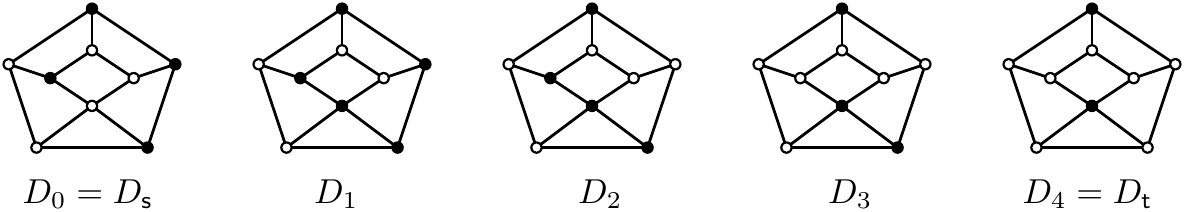}
    \caption{Reconfiguration sequence between two dominating sets $D_\source$ and $D_\target$ under the {\sf TAR}(5) rule; the dominating sets are depicted by the set of black vertices.}
    \label{fig:sequence}
\end{figure}

Let $G$ be a graph. We denote by $\Gamma(G)$ the maximum size of a dominating set which is minimal by inclusion. Determining upper bounds on $k$ that guarantee that the $k$-reconfiguration graph $\mathcal{R}_k(G)$ is connected has received a lot of attention. Haas and Seyffarth proved in~\cite{HaasS14} that being reconfigurable is not a monotone property, which means that if $\mathcal{R}_k(G)$ is connected then $\mathcal{R}_{k+1}(G)$ is not necessarily connected. Indeed, let us denote by $K_{1,n}$ the star graph on $n+1$ vertices, and note that $\Gamma(K_{1,n}) = n$. They observed that, for every $n \ge 3$, $\mathcal{R}_k(K_{1,n})$ is connected if $1 \le k \le n-1$. But $\mathcal{R}_n(K_{1,n})$ is not connected since the dominating set of size $n$ which contains all the degree-one vertices is frozen, i.e. it is an isolated vertex in $\mathcal{R}_n(K_{1,n})$. They then asked what is the smallest integer $d_0$ such that $\mathcal{R}_k(G)$ is connected, for any $k \ge d_0$. They proved the following:

\begin{lemma}[\cite{HaasS14}]
Let $G$ be a graph. If $k > \Gamma(G)$ and $\mathcal{R}_k(G)$ is connected, then $\mathcal{R}_{k+1}(G)$ is connected.
\end{lemma}

Moreover, they proved that if $G$ has at least two independent edges, then $d_0 \le \min\{n-1, \Gamma(G)+\gamma(G)\}$, $\gamma(G)$ being the size of a minimum dominating set of $G$. They also showed that this value can be lowered to $\Gamma(G)+1$ if $G$ is bipartite or a chordal graph. This result is tight since $K_{1,n}$ is bipartite and chordal and $\mathcal{R}_k(K_{1,n})$ is not connected. They asked if this result can be generalized to any graph. Suzuki et al.~\cite{SuzukiMN16} answered negatively this question by constructing an infinite family of graphs for which $\mathcal{R}_{\Gamma(G)+1}(G)$ is not connected. Mynhardt et al.~\cite{Mynhardt19} improved this result by constructing two infinite families of graphs:

\begin{itemize}
    \item the first construction provides graphs with arbitrary $\Gamma \ge 3$, arbitrary domination number in the range $2 \le \gamma \le \Gamma$ such that $d_0 = \Gamma + \gamma -1$
    \item the second one gives graphs with arbitrary $\Gamma \ge 3$, arbitrary domination number in the range $1 \le \gamma \le \Gamma-1$ for which $d_0= \Gamma + \gamma$. For $\gamma \ge 2$, this is the first construction of graphs with $d_0= \Gamma + \gamma$. 
\end{itemize}

On the positive side, Haas and Seyffarth~\cite{Haas17} proved that if $k = \Gamma(G)+\alpha(G)-1$ (where $\alpha(G)$ is the size of a maximum independent set of $G$), then $\mathcal{R}_k(G)$ is connected. To obtain this result, they proved that all the independent dominating sets of $G$ are in the same connected component of $\mathcal{R}_{\Gamma(G)+1}(G)$. Recall that if $G$ has at least two independent edges, then $d_0 \le \min\{n-1, \Gamma+\gamma(G)\}$. It implies that the aforementioned value of $d_0$ obtained by Mynhardt et al. in \cite{Mynhardt19} is the best we can hope for in the general case since $d_0 \le \min\{\Gamma(G)+\gamma(G), 2\Gamma(G)-1\}$ holds for any graph $G$. 

Haddadan et al.~\cite{Haddadan16} studied the algorithmic complexity of the problem. They proved that, given a graph $G$, two dominating sets $D_\source$ and $D_\target$ of $G$ and an integer $k \ge \max\{|D_\source|,|D_\target|\}$, it is PSPACE-complete to decide whether there exists a path in $\mathcal{R}_k(G)$ between $D_\source$ and $D_\target$. Actually, this problem remains PSPACE-complete even restricted to bipartite graphs or split graphs. On the other hand, they proved that this problem can be decided in linear time if the input graph is a tree, an interval graph or a cograph. 

\medskip
Mouawad et al.~\cite{Mouawad16} studied the problem from a parameterized point of view. They proved that this problem is W[2]-hard parameterized by $k+\ell$, where $k$ is the threshold and $\ell$ the size of the desired reconfiguration sequence. On the positive side, Lokshtanov et al.~\cite{Lokshtanov18} gave an FPT algorithm parameterized by $k$ for graphs excluding $K_{d,d}$ as a subgraph, for any constant $d$. Finally, Blanch\'e et al.~\cite{Blanche20} studied the complexity and parameterized complexity of an optimization variant originally introduced by Ito et al.~\cite{Ito19} for the independent set reconfiguration problem.

\paragraph{Our contribution.}
Let $G=(V,E)$ be a graph on $n$ vertices.
In Section \ref{sec:UB}, we show that if $k=\Gamma(G)+\alpha(G)-1$, then $\mathcal{R}_k(G)$ has linear diameter, improving a previous result of Haas and Seyffarth~\cite{Haas17} which only proved that $\mathcal{R}_k(G)$ is connected but did not give any bound on the diameter\footnote{Their induction based proof does not provide a linear diameter.}. Note that the proof is algorithmic, and outputs such a transformation in polynomial time.
It contrasts in particular with a result of Suzuki et al.~\cite{SuzukiMN16} who provided an infinite family of graphs $G_n$ of linear size for which $\mathcal{R}_{\gamma+1}(G)$ has diameter $\Omega(2^n)$. 

In Section~\ref{sec:minorUB}, we give some threshold that guarantee that $\mathcal{R}_k(G)$ is connected and has linear diameter for some "minor sparse classes"\footnote{For a formal definition, we refer the reader to Section~\ref{sec:minorUB}.}. In particular, we prove that $\mathcal{R}_k(G)$ is connected and has linear diameter for $K_{\ell}$-minor free graphs as long as $k \ge \Gamma(G)+O(\ell \sqrt{\log \ell})$. In the particular case of planar graphs, it actually holds as long as $k \ge \Gamma(G)+3$. The proof is algorithmic, and provides linear transformations in polynomial time. We know that there exist planar graphs for which $k \ge \Gamma(G)+2$ is necessary~\cite{SuzukiMN16}. We conjecture the following:

\begin{conjecture}\label{conjecture-planar}
For every planar graph $G$, $\mathcal{R}_{\Gamma(G)+2}(G)$ is connected.
\end{conjecture}

For $K_\ell$-minor free graphs, the gap between the lower and upper bound is not completely closed since the only lower bound we know is $\Gamma(G)+\ell-4$, which is the lower bound for graphs of treewidth at most $\ell-2$ which will be discussed in the next paragraph (graphs of treewidth at most $\ell-2$ are $K_\ell$-minor free). Our argument for  $K_\ell$-minor free graphs is based on their average degree, and then we cannot improve the term $\Gamma(G)+O(\ell \sqrt{\log \ell})$ with our proof technique.

Finally, in Section~\ref{sec:twUB} we give a sharper upper bound for bounded treewidth graphs. We prove that $\mathcal{R}_k(G)$ is connected for $k=\Gamma(G)+tw(G)+1$, and has linear diameter. Again our results are algorithmic as long as the tree decomposition is given. Since a tree-decomposition of width $r$ can be found in time $2^{O(k^3)} \cdot n$~\cite{Bodlaender96}, our results provide an FPT algorithm parameterized by the treewidth that outputs a linear transformation between any two dominating sets as long as $k \ge \Gamma(G)+tw(G)+1$. 

We claim that this bound is tight up to an additive constant factor. Mynhardt et al.~\cite{Mynhardt19} constructed an infinite family of graphs $G_{\ell,r}$ (with $\ell \ge 3$ and $1 \le r \le \ell-1$) for which $2\Gamma(G)-1$ tokens are necessary to guarantee the connectivity of the reconfiguration graph. Let us describe their construction when $r=\ell-1$. The graph $G_{\ell, \ell-1}$ contains $\ell-1$ cliques $C_1,C_2, \ldots, C_{\ell-1}$ called {\em inner cliques}, each of size $\ell$. We denote by $c^j_i$ the $j$-th vertex of the clique $C_i$. We then add a new clique $C_0$ of size $\ell$, called the {\em outer clique} and we add a new vertex $u_0$ adjacent to all the vertices of $C_0$ (hence, $C_0$ can be seen as a clique of size $\ell+1$). For every $1 \le i \le \ell-1$ and for every $1 \le j \le \ell$, we add an edge between $c^j_i$ and $c^j_0$. This completes the construction of $G_{\ell,\ell-1}$ (see Figure~\ref{fig:lower-bound} for an example). Mynhardt et al.~\cite{Mynhardt19} showed that $\Gamma(G_{\ell,\ell-1})=\ell$.  

They moreover show that $\mathcal{R}_{2\ell-2}(G_{\ell,\ell-1})$ is not connected. One can prove easily (see Section~\ref{sec:prelim}) that $G_{\ell,\ell-1}$ has treewidth $\ell$. So $\mathcal{R}_{\Gamma(G)+tw(G)-2}$ is not necessarily connected. So our function of the treewidth is tight up to an additive constant factor. 
The pathwidth of $G_{\ell,\ell-1}$ is at most $2\ell-1$. However, it is not clear if and how we can obtain a better upper bound for bounded pathwidth graphs. To sum up $\mathcal{R}_k(G)$ is not necessarily connected if $k < \Gamma(G)+pw(G)/2+O(1)$ and is connected if $k>\Gamma(G)+pw(G)+1$. We were not able to close this gap and left it as an open problem.

\begin{figure}[bt]
    \centering
    \includegraphics[scale=.6]{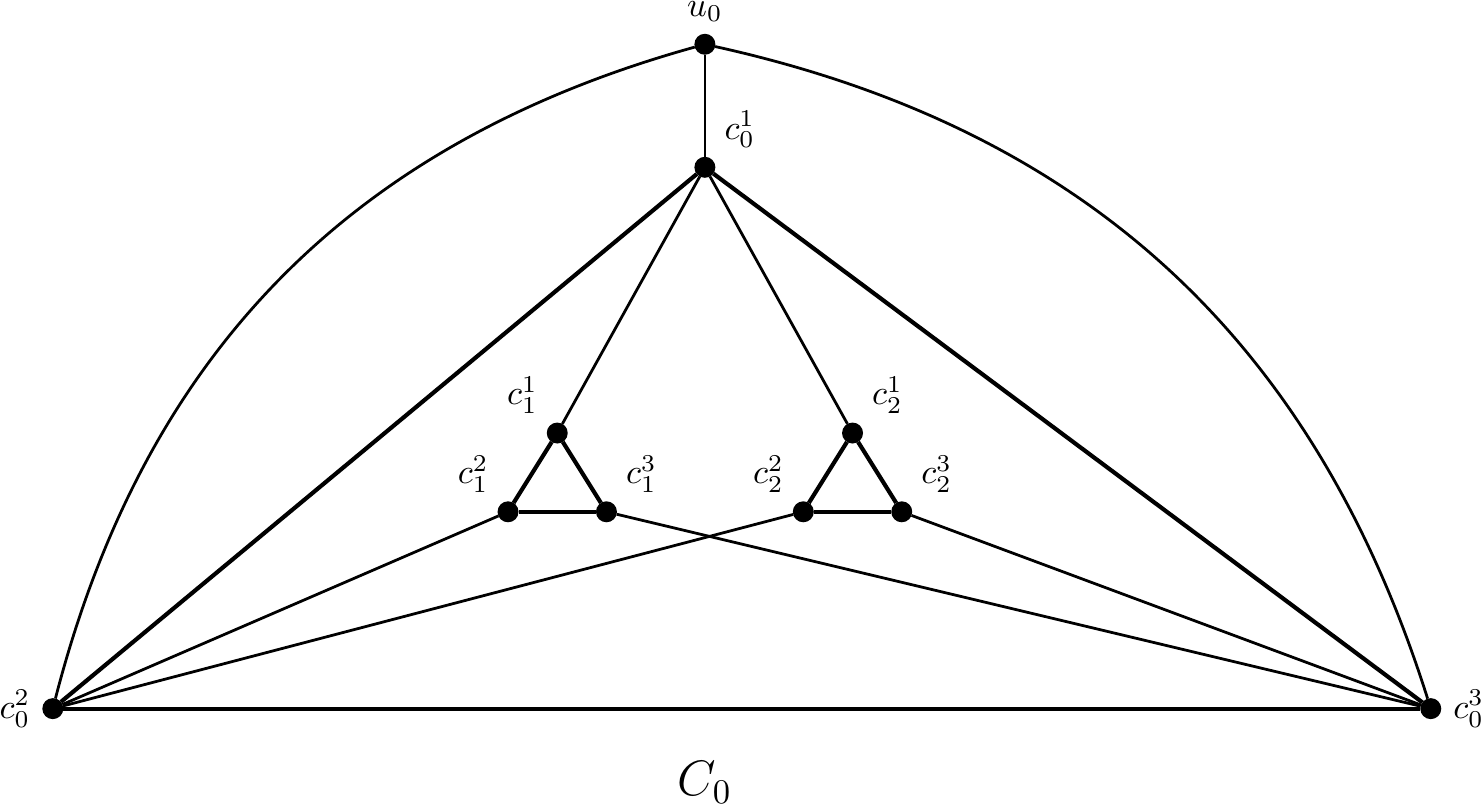}
    \caption{The graph $G_{3,2}$}
    \label{fig:lower-bound}
\end{figure}

\section{Preliminaries}\label{sec:prelim}
All along the paper, every graph we consider is finite and simple. Let $G=(V,E)$ be a graph. When there is no ambiguity on the graph $G$, $V$ denotes the vertex set of $G$, $E$ its set of edges, $n$ its order and $m$ its size. 

Given a subset of vertices $S \subseteq V$, we denote by $G[S]$ the subgraph of $G$ induced by $S$. More precisely, the vertex set of $G[S]$ is $S$, and its edge set is the subset of edges of $G$ with both endpoints in $S$. 

An \emph{edge contraction} is an operation which removes an edge from a graph while simultaneously merging the two vertices it used to connect (the resulting new vertices is adjacent to a vertex $v$ if and only if at least one endpoint of the edge was incident to $v$). A graph $H$ is a \emph{minor} of $G$ if a graph isomorphic to $H$ can be obtained from $G$ by contracting some edges, deleting some edges, and deleting some isolated vertices.

Given a vertex $v \in V$, $N(v)$ denotes the \emph{neighborhood} of $v$, i.e. the set $\{ u \in V ~ | ~ uv \in E \}$. We denote by $N[v]$ the {\em closed neighborhood} of $v$, that is the set $N(v) \cup \{v\}$.

A \emph{dominating set} $D$ of $G$ is a subset of $V$ such that for any $v\in V$, $v\in D$ or there exists $u\in D$ such that $uv\in E$. An \emph{inclusion-wise minimal dominating set} of $G$ is a dominating set $D$ of $G$ such that for any $v\in D$, $D\setminus v$ is not a dominating set of $G$. A \emph{minimum dominating set} of $G$ is a dominating set $D$ of $G$ such that $|D|$ is minimal with this property. The maximum size of a minimal dominating set of $G$ is denoted by $\Gamma(G)$. We say that a set $X\subseteq V$ \emph{dominates} another set $Y\subseteq V$ if for any $v\in Y$, there exists $u\in X$ such that $uv\in E$.

An {\em independent set} (or {\em stable set}) of $G$ is a subset $S \subseteq V$ of pairwise non-adjacent vertices, i.e. for any pair of vertices $u,v \in S$, $uv \not\in E$. An {\em inclusion-wise maximal independent set} $S$ is an independent set such that for any $v \in V \setminus S$, there exists $u \in N(v)$ such that $u \in S$. A {\em maximum independent set} of $G$ is an independent set $S$ such that $|S|$ is maximal. We denote by $\alpha(G)$ the \emph{independence number} of $G$, that is the size of a maximum independent set. Computing a maximum independent set of a given graph $G$ is a classical NP-complete problem, while computing a maximal one can trivially be done in linear time by a greedy algorithm. Moreover, given an independent set $S'$ which is not maximal, one can greedily complete into a maximal independent set $S$ such that $S' \subseteq S$. In particular, if there exist two vertices $u$ and $v$ such that $uv \not\in E$, then there exists a maximal independent set of $G$ which contains both $u$ and $v$. Obviously, this is also true when $S'$ is reduced to a single vertex. We will use this fact in the proof of Theorem~\ref{thm:upper-bound}, as well as the following well-known observation:

\begin{observation} \label{obs:maxIS-minDS}
Let $G=(V,E)$ be a graph, and $S \subseteq V$ be an inclusion-wise maximal independent set of $G$. Then, $S$ is an inclusion-wise minimal dominating set.
\end{observation}

\begin{proof}
Let $u \in V$ be a vertex. If $u \in S$, $u$ is dominated by itself. Otherwise, there exists $v \in N(u) \cap S$ since $S$ is maximal. Hence, $u$ is dominated by $v$. Moreover, by definition of an independent set, we have $N(S\setminus u)$ does not contain $u$ for every vertex $u \in S$. Therefore, $u$ is not dominated in $S \setminus \{u\}$ and thus $S$ is a minimal dominating set of $G$. 
\end{proof}

Note that Observation~\ref{obs:maxIS-minDS} implies that any inclusion-wise maximal independent set $S$ of $G$ satisfies $|S| \le \alpha(G) \le \Gamma(G)$.   
In the remaining, we often refer to inclusion-wise minimal dominating sets (respectively inclusion-wise maximal independent sets) as \emph{minimal dominating sets} (respectively \emph{maximal independent sets}) by abuse of language.

A \emph{tree} is a connected graph that contains no cycle. Given a graph $G=(V,E)$, a \emph{tree decomposition} of $G$ is a pair $(X,T)$ where $X$ is a set of subsets of $V$ called \emph{bags} and $T$ is a tree whose vertices are the bags of $X$, and that satisfies:

\begin{itemize}
    \item For any vertex $v\in V$, $v$ belongs to at least one bag of $X$
    \item For any edge $uv\in E$, there exists a bag that contains both $u$ and $v$
    \item For any vertex $v\in V$, the set of bags containing $v$ forms a subtree of $T$.
\end{itemize}

The minimum, over all the possible tree decompositions of $G$, of the maximum size of a bag, to which we subtract $1$, is called the \emph{treewidth} of $G$ and is denoted by $tw(G)$.  A \emph{path decomposition} is a tree decomposition such that $T$ is a path. The minimum, over all the path decompositions of $G$, of the maximum size of a bag minus $1$ is the \emph{pathwidth} of $G$, denoted by $pw(G)$.

\paragraph{Pathwidth and treewidth of $G_{\ell,\ell-1}$.}
In the introduction, we claimed that $G_{\ell,\ell-1}$ has treewidth $\ell$ and pathwidth at most $2\ell-1$. For completeness, we prove it here.

\begin{claim} \label{claim:tw-lower-bound}
The graph $G_{\ell,\ell-1}$ has treewidth $\ell$.
\end{claim}

\begin{proofclaim}
First, observe that $tw(G_{\ell,\ell-1}) \ge \ell$ since $G[C_0 \cup \{u_0\}]$ is a clique of size $\ell+1$. 

Let us now give a tree decomposition of $G_{\ell,\ell-1}$ of width $\ell$. We first create a ``central'' bag $B_0$ containing all the vertices of $C_0$ and the vertex $u_0$. For each inner clique $C_i$ with $1 \le i \le \ell-1$, we attach to $B_0$ a path $B_i^1B_i^2 \cdots B_i^{\ell}$ where $B_i^j$ contains the vertices $(C_0 \setminus \bigcup_{k=0}^{j-1} c_0^k) \cup \bigcup_{k=1}^j c_i^k$ (see Figure~\ref{fig:lower-bound-tw} for an example). Observe that for any $1 \le i \le \ell-1$, the bag $B_i^\ell$ contains all the vertices of $C_i$. And the bag $B_i^j$ contains both $c_0^j$ and $c_i^j$. Hence, each edge is contained in at least one bag. For every $1 \le j \le \ell$, the vertex $c_0^j$ is contained in the bags $B_0 \cup \bigcup_{i=1}^{\ell-1} \bigcup_{k=1}^j B_i^k$. And for every $1 \le i \le \ell-1$ and every $1 \le j \le \ell$, the vertex $c_i^j$ is contained in $B_i^1, B_i^2, \ldots, B_i^j$. It follows that for every vertex $u \in V(G_{\ell,\ell-1})$ the set of bag containing $u$ induces a connected subtree. Finally, one can easily check that each bag contains exactly $\ell+1$ vertices. Hence, this decomposition indeed is a tree decomposition of $G_{\ell,\ell-1}$ of width $\ell$ and the conclusions follows.
\end{proofclaim}

\begin{figure}[bt]
    \centering
    \includegraphics[scale=0.9]{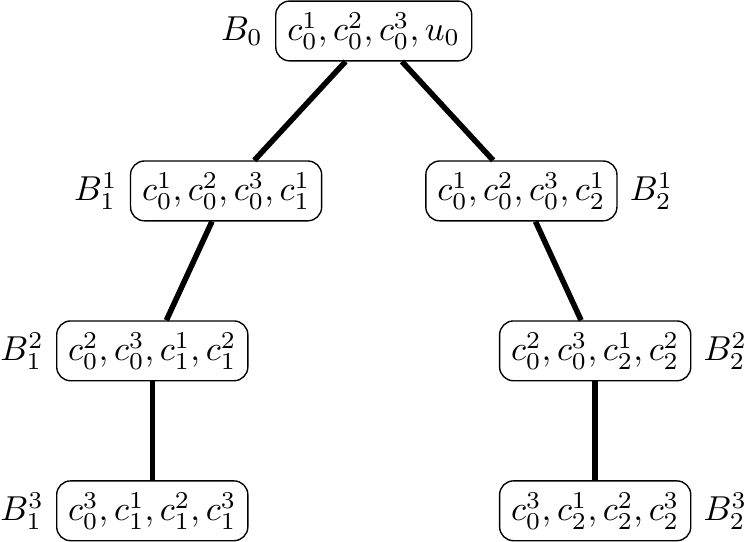}
    \caption{Tree decomposition of $G_{3,2}$ of width $tw(G_{3,2})$.}
    \label{fig:lower-bound-tw}
\end{figure}

\begin{claim}\label{claim:pw-lower-bound}
The pathwidth of $G_{\ell,\ell-1}$ is at most $2\ell-1$.
\end{claim}

\begin{proofclaim}
We give a path decomposition of width at most $2\ell-1$ of $G_{\ell,\ell-1}$. We first create a bag $B_0$ which contains $C_0 \cup \{u_0\}$. For every $1 \le i \le \ell-1$, we create a bag $B_i = C_0 \cup C_i$ such that $B_1B_2 \ldots B_{\ell-1}$ induces a path. One can easily check that it is a path decomposition of width $2\ell-1$ of $G_{\ell,\ell-1}$.
\end{proofclaim}

\section{General upper bound}\label{sec:UB}
Let $G$ be a graph. All along the section $k=\Gamma(G)+\alpha(G)-1$.
Haas and Seyffarth showed that $\mathcal{R}_k(G)$ is connected~\cite{Haas17}. However, they do not explicit the diameter of the reconfiguration graph and their induction based proof does not give a linear diameter. We propose a new proof of the same result that moreover implies that the reconfiguration graph has linear diameter. Note that our proof is constructive and provides an algorithm that construct a path between two given dominating sets of size at most $k$ of $G$.

\begin{observation}\label{obs:IS-DS}
Let $D$ be a minimal dominating set of $G$, and let $S$ be a maximal independent set of $G$ such that $D \cap S \neq \emptyset$. Then, there exists a {\sf TAR}(k)-reconfiguration sequence between $D$ and $S$ of length at most $|D|+\alpha(G)-2$.
\end{observation}

\begin{proof}
Recall that since $S$ is a maximal independent set, $|S| \le \alpha(G) \le \Gamma(G)$. We first add to $D$ each vertex in $S \setminus D$ one by one. Note that there are at most $\alpha(G)-1$ such vertices. We thus obtain the set $D' = D \cup S$. We then remove one by one each vertex in $D \setminus S$. There are at most $|D|-1$ such vertices since $S \cap D \neq \emptyset$. Each intermediate solution is indeed a dominating set since it either contains $D$ or $S$ which are both dominating sets. Moreover, each solution is of size at most $|D'| \le |D|+|S|-1 \le k$. 
\end{proof}

\begin{theorem} \label{thm:upper-bound}
Let $G=(V,E$) be a graph on $n$ vertices. If $k  = \Gamma(G) + \alpha(G) - 1$ then $\mathcal{R}_k(G)$ has diameter at most $10n$.
\end{theorem}

\begin{proof}
Let $D_1$ and $D_2$ be two dominating sets, both of size at most $k$. Free to remove at most $2\cdot (\Gamma(G) + \alpha(G) - 2)$ vertices in total, one can assume without loss of generality that $D_1$ and $D_2$ are both inclusion-wise minimal dominating sets of $G$. Hence $|D_1| \le \Gamma(G)$ and $|D_2| \le \Gamma(G)$. We outline a path between $D_1$ and $D_2$ in $\mathcal{R}_k(G)$. The next claim deals with the case where $D_1$ and $D_2$ have a non-empty intersection.

\begin{claim}\label{claim:non-empty-intersection}
If $D_1 \cap D_2 \neq \emptyset$ then there exists a reconfiguration sequence from $D_1$ to $D_2$ of length at most $2 \cdot (\alpha(G)+\Gamma(G)-2)$.
\end{claim}

\begin{proofclaim}
Let $x$ be a vertex that belongs to both $D_1$ and $D_2$. One first constructs greedily (and thus in polynomial-time) a maximal independent set $S$ of $G$ which contains $x$ (which is then of size at most $\alpha(G)$). By Observation~\ref{obs:IS-DS}, one can transform $D_1$ into $S$ under the {\sf TAR}($k$) rule. And the length of the reconfiguration sequence is at most $\Gamma(G)+\alpha(G)-2$.
Similarly, there exists a reconfiguration sequence of length at most $\Gamma(G)+\alpha(G)-2$ from $D_2$ to $S$. By combining these two transformations, we obtain a reconfiguration sequence between $D_1$ and $D_2$ of length at most $2 \cdot (\alpha(G)+\Gamma(G)-2)$, as desired.
\end{proofclaim}

In the remaining of this proof, we assume that $D_1 \cap D_2 = \emptyset$ otherwise we can directly conclude by Claim~\ref{claim:non-empty-intersection}. If there exist $u_i \in D_1$ and $v_j \in D_2$ such that the set $D' = (D_1 \setminus \{u_i\}) \cup \{v_j\}$ is a dominating set of $G$, then we can conclude by Claim~\ref{claim:non-empty-intersection} since $D' \cap D_2 \neq \emptyset$ and $D'$ can be obtained from $D_1$ in two steps. Suppose now that $D' = (D_1 \setminus \{u_i\}) \cup \{v_j\}$ is not a dominating set of $G$. This means that $u_i$ is adjacent to a vertex $x$ with no neighbors in $(D_1 \setminus \{u_i\}) \cup \{v_j\}$. Hence, there exists a maximal independent set $S_1$ of $G$ which contains both $x$ and a vertex $u_k \in D_1 \setminus \{u_i\}$. Similarly, there exists a maximal independent set $S_2$ which contains both $x$ and $v_j$. By Observation~\ref{obs:IS-DS}, there exists a reconfiguration sequence of length at most $\Gamma(G)+\alpha(G)-2$ between $S_1$ (respectively $S_2$) and $D_1$ (respectively $D_2$) under the {\sf TAR}($k$) rule. 
Finally, since $S_1$ and $S_2$ intersect, we can again use Observation~\ref{obs:IS-DS} that ensures that there exists a transformation from $S_1$ to $S_2$ of length at most $2\alpha(G)-2$.

Hence, we obtain a {\sf TAR}($k$)-reconfiguration sequence from $D_1$ to $D_2$ of length at most $4 \cdot (\Gamma(G) + \alpha(G)-2) + 2 \cdot (\alpha(G)-1) < 10n$. 
\end{proof}

\section{\texorpdfstring{$H$}{H}-minor free graphs}\label{sec:minorUB}

In this section, we will prove some better bounds on $k$ for minor-free graphs.
We say that a graph is \emph{$d$-minor sparse} if all its bipartite minors have average degree less than $d$. Note that it is equivalent to say that the ratio between the number of edges and the number of vertices of any bipartite minor of $G$ is strictly less than $\frac{d}{2}$.

\begin{lemma}\label{lem:contradiction}
Let $G$ be a $d$-minor sparse graph.
Let $A$ and $B$ be two dominating sets of $G$ such that $|A|=|B|$ and $|B\setminus A|\geq d$. Then, there exists a vertex $a \in A\setminus B$ and a set $S \subset B\setminus A$ with $|S|=d-1$ such that $(A \cup S) \setminus \{a\}$ is a dominating set of $G$.
\end{lemma}
\begin{proof}
We prove it by contradiction.
For every $a_i\in A\setminus B$, let $S_{i,1}$ be a subset of $B\setminus A$ of size $d-1$. Let $x_{i,1}$ be a vertex that is only dominated by $a_i$ in $A$ and not dominated by $S_{i,1}$ in $B$ (such a vertex must exist otherwise the conclusion follows). Note that this vertex might be a vertex of $A$ or of $B$.
Let $b_{i,1}$ be a vertex of $(B\setminus A)\setminus S_{i,1}$ that dominates $x_{i,1}$. This vertex exists since $B$ is a dominating set and $x_{i,1}$ is only dominated by $a_i$ in $A$. Now, for every $2\leq j\leq d$, we define recursively the set $S_{i,j}$ as a subset of size $d-1$ of $B\setminus A$ containing $\{b_{i,1},\ldots,b_{i,j-1}\}$. We let $x_{i,j}$ be a vertex only dominated by $a_i$ in $A$ that is not dominated by $S_{i,j}$ in $B$, and $b_{i,j}$ be a vertex of $(B\setminus A)\setminus S_{i,j}$ that dominates $x_{i,j}$. Note that, for every $j$, since $x_{i,j}$ is incident to $b_{i,j}$ and not to  $S_{i,j}$, $b_{i,j} \notin \{ b_{i,1},\ldots,b_{i,j-1} \}$. In particular, $B_i:=\{b_{i,1},\ldots,b_{i,d}\}$ has size exactly $d$. Note that $B_i \subseteq B \setminus A$. The construction of the set $B_i$ is illustrated in Figure \ref{fig:minor}.

\begin{figure}[bt]
	\centering
	\scalebox{0.8}{
	\begin{tikzpicture}[scale=0.2]
    
    \draw[fill=gray!20] (0,16) ellipse (25cm and 3cm);
    \draw[fill=gray!20] (0,0) ellipse (25cm and 3cm);
    \node[right] at (25,16) {$B\setminus A$};
    \node[right] at (25,0) {$A\setminus B$};
    
    \node[draw,thick,circle,fill=black,inner sep=1.5pt,label=left:$a_i$] (ai) at (5,0) {};
    \node[draw,thick,circle,fill=black,inner sep=1.5pt,label=left:$a_j$] (aj) at (-15,0) {};
    \node[draw,thick,circle,fill=black,inner sep=1.5pt,label=right:$x_{i,1}$] (x11) at (-5,8) {};
    \node[draw,thick,circle,fill=black,inner sep=1.5pt,label=left:$b_{i,1}$] (b11) at (-5,16) {};
    \node[draw,thick,circle,fill=black,inner sep=1.5pt,label=right:$x_{i,2}$] (x12) at (5,8) {};
    \node[draw,thick,circle,fill=black,inner sep=1.5pt,label=left:$b_{i,2}$] (b12) at (5,16) {};
    \node[draw,thick,circle,fill=black,inner sep=1.5pt,label=right:$x_{i,3}$] (x13) at (15,8) {};
    \node[draw,thick,circle,fill=black,inner sep=1.5pt,label=left:$b_{i,3}$] (b13) at (15,16) {};
    
    \node (dots) at (18,16) {$\ldots$};
    \node (dots) at (18,0) {$\ldots$};
    \node (dots) at (22,8) {$\ldots$};
    
    \draw[dashed] (aj) -- (x11);
    \draw[dashed] (aj) -- (x12);
    \draw[dashed] (aj) -- (x13);

    \draw[thick] (x11)-- (b11);
    \draw[decorate,decoration={zigzag}] (ai) -- (x11);
    \draw[thick] (x12)-- (b12);
    \draw[decorate,decoration={zigzag}] (ai) -- (x12);
    \draw[thick] (x13)-- (b13);
    \draw[decorate,decoration={zigzag}] (ai) -- (x13);
    
    \draw[dashed] (x12)-- (b11);
    \draw[dashed] (x13)-- (b11);
    \draw[dashed] (x13)-- (b12);
    
	\end{tikzpicture}}
    \caption{The set $B_i$. The dotted lines represent the non-edges, and the zigzags represent the edges that are contracted in $G'$.}
	\label{fig:minor}
\end{figure}
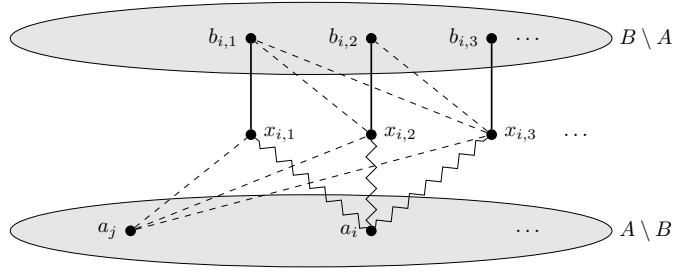

Let us construct a minor $G'$ of $G$ of density at least $d$. For every $a_i\in A\setminus B$, we contract the edges $a_ix_{i,j}$ for any $j$ such that $x_{i,j}\not \in B\setminus A$ and $x_{i,j}\not \in A\setminus B$. (In particular, if $x_{i,j} \in B$, we do not contract the edge). By abuse of notations, we still denote by $a_i$ the resulting vertex. Note that the vertices $x_{i,j}$ are pairwise disjoint. If $x_{i,j}=x_{i',j'}$ then, since $x_{i,j}$ is only dominated by $a_i$ and $x_{i',j'}$ by $a_i'$, we must have $a_i=a_i'$. And by construction in the previous paragraph, $x_{i,j} \ne x_{i,j'}$ if $j \ne j'$. So the contractions defined above are well defined. Moreover, the size of $A\setminus B$ is left unchanged. Similarly the size of $B \setminus A$ is not modified. We finally remove from the graph any vertex which is not in $(A\setminus B)\cup (B\setminus A)$, and any edge internal to $A\setminus B$ and to $B\setminus A$. The resulting graph $G'$ is a minor of $G$ and is bipartite. 

For every $i,j$, $x_{i,j}$ is adjacent to every vertex of $B_i$ in $G$. Thus, $a_i$ is adjacent to every vertex of $B_i$ in $G'$. Therefore, for any $a_i\in A\setminus B$, $a_i$ has degree at least $d$ in $G'$. Thus, there are at least $d \cdot |A\setminus B|$ edges in $G'$. Since $G'$ has $|A\setminus B|+|B\setminus A| = 2|A\setminus B|$ vertices, it contradicts the fact that $G$ is a $d$-minor sparse graph.
\end{proof}

\begin{lemma}\label{lem:minorUB}
Let $G$ be a $d$-minor sparse graph.
If $k=\Gamma(G)+d-1$, then $\mathcal{R}_k(G)$ is connected and the diameter of $\mathcal{R}_k(G)$ is at most $2\Gamma(G)\cdot(d-1)+2(\Gamma(G)-1)$.
\end{lemma}

\begin{proof}
Firstly, if $d > \Gamma(G)$, then the result follows from Theorem~\ref{thm:upper-bound}. So we assume $d \leq \Gamma(G)$. We proceed by induction on $|\target \setminus D_\source|$.
Let $D_\source$ and $D_\target$ be two dominating sets of $G$ of size at most $k$. Since $\Gamma(G)$ is the maximum size of a dominating set minimal by inclusion, we can add or remove vertices from $D_\source$ and $D_\target$ so that $D_\source$ and $D_\target$ both have size exactly $\Gamma(G)$, keeping dominating sets. Note that by assumption, we need to remove or add at most $2(\Gamma(G)-1)$ vertices in total. So from now on, we assume that $|D_\source| = |D_\target| = \Gamma(G)$. Let us show that there is a path from $D_\source$ to $D_\target$ in $\mathcal{R}_k(G)$ of length at most $2|D_\target \setminus D_\source|\cdot(d-1)$. Since $|D_\target \setminus D_\source|\leq \Gamma(G)$, and by taking into account the $2(\Gamma(G)-1)$ vertices eventually initially added or removed, this will give the expected result. We proceed by induction on $|D_\target \setminus D_\source|$.

If $|D_\target \setminus D_\source|\leq d-1$ then, since $|D_\source| = \Gamma(G)$, we have $|D_\source \cup D_\target|\leq \Gamma(G)+d-1$. Thus, we can simply add all the vertices of $D_\target \setminus D_\source$ to $D_\source$ and then remove the vertices of $D_\source \setminus D_\target$. We thus obtain a path from $D_\source$ to $D_\target$ in $\mathcal{R}_k(G)$ of length at most $2d-2\leq 2| D_\target \setminus D_\source|\cdot(d-1)$.

Assume now that $|D_\target \setminus D_\source|\geq d$. By Lemma~\ref{lem:contradiction}, there exists a vertex $v \in D_\source \setminus D_\target$ and a set $S \subset D_\target \setminus D_\source$ with $|S|=d-1$ such that $D_\source':= (D_\source \cup S) \setminus \{v\}$ is a dominating set of $G$. Let $D_\source''$ be any dominating set of size exactly $\Gamma(G)$ obtained by removing vertices of $D_\source'$, i.e. such that $D_\source'' \subseteq D_\source'$. Since $|S|=d-1$ and $|D_\source|=\Gamma(G)$, the transformation that consists in adding every vertex of $S$ to $D_\source$ and then removing $v$ and every vertex of $D_\source'\setminus D_\source''$ is a path from $D_\source$ to $D_\source''$ in $\mathcal{R}_k(G)$. Moreover, $|D_\source'|= \Gamma(G)+d-2$. Thus, this path has length $2d-2$.

We have $D_\source':= (D_\source \cup S) \setminus \{v\}$ where $v\in D_\source\setminus D_\target$ and $S\subset D_\target\setminus D_\source$ with $|S|=d-1$. Thus, $|D_\target \setminus D_\source '|=|D_\target \setminus D_\source|-d+1$. Since $D_\source''\subseteq D_\source'$ and $|D_\source'\setminus D_\source''|\leq d-2$, it gives $|D_\target \setminus D_\source''|\leq |D_\target \setminus D_\source|-1$. By induction hypothesis, there exists a path from $D_\source''$ to $D_\target$ in $\mathcal{R}_k(G)$ of length at most $|D_\target \setminus D_\source''|\cdot(2d-2)$. The concatenation of the two paths gives a path from $D_\source$ to $D_\target$ in $\mathcal{R}_k(G)$ of length at most $2|D_\target \setminus D_\source|\cdot(d-1)$. This concludes the proof.
\end{proof}

 Let us now state two immediate corollaries of Lemma~\ref{lem:minorUB}: 
 
\begin{corollary}\label{coro:minor}
Let $G$ be a graph. Then, we have the following:
\begin{itemize}
    \item if $G$ is planar, then $\mathcal{R}_{k}(G)$ is connected and has linear diameter for every $k \ge \Gamma(G)+3$.
    \item if $G$ is $K_\ell$-minor free, then there exists a constant $C$ such that $\mathcal{R}_k(G)$ is connected and has linear diameter for every $k \ge \Gamma(G)+C \ell \sqrt{\log_2 \ell}$.
\end{itemize}
\end{corollary}

\begin{proof}
Every minor of a planar graph is planar. Moreover every bipartite planar graph has at most $2n-4$ edges. Thus every planar graph is a $4$-minor sparse graph and the first point follows from Lemma~\ref{lem:minorUB}.

A result of Thomason~\cite{thomason84} (improving a result of Mader~\cite{Mader68}) ensures that the average degree of a  $K_\ell$-minor free graph is at most $0.265 \cdot\ell \sqrt{\log_2 \ell}(1 + o(1))$. In particular, there exists a constant $C$ such that, for every $\ell$ and every $K_\ell$-minor free graph $G$, the average degree of $G$ is at most $C \ell \sqrt{\log_2 \ell}$. Thus $G$ is $C \ell \sqrt{\log_2 \ell}$-minor sparse and the second point follows from  Lemma~\ref{lem:minorUB}.
\end{proof}

We were not able to find an example where $\Gamma(G)+3$ is needed for planar graphs. We also know that $\Gamma(G)+1$ is not enough. Indeed, Suzuki et al \cite{SuzukiMN16} gave an example of a planar graph $G$ for which $\mathcal{R}_{\Gamma(G)+1}(G)$ is disconnected. The graph $G$ is given in Figure \ref{fig:planar}.

\begin{figure}
	\centering
	\scalebox{0.8}{
	\begin{tikzpicture}[scale=0.8]
    
    \node[draw,thick,circle,fill=black,inner sep=1.5pt] (a1) at (0.5,0.5) {};
    \node[draw,thick,circle,fill=black,inner sep=1.5pt] (a2) at (5.5,0.5) {};
    \node[draw,thick,circle,fill=black,inner sep=1.5pt] (a3) at (3,5.5) {};
    
    \node[draw,thick,circle,fill=white,inner sep=1.5pt] (b1) at (1.5,1.3) {};
    \node[draw,thick,circle,fill=white,inner sep=1.5pt] (b2) at (4.5,1.3) {};
    \node[draw,thick,circle,fill=white,inner sep=1.5pt] (b3) at (3,4) {};
    
    \node[draw,thick,circle,fill=black,inner sep=1.5pt] (c1) at (2.5,2) {};
    \node[draw,thick,circle,fill=black,inner sep=1.5pt] (c2) at (3.5,2) {};
    \node[draw,thick,circle,fill=black,inner sep=1.5pt] (c3) at (3,3) {};

    \draw[thick] (a1)-- (a2) -- (a3) -- (a1);
    \draw[thick] (b1)-- (b2) -- (b3) -- (b1);
    \draw[thick] (c1)-- (c2) -- (c3) -- (c1);
    \draw[thick] (a1)-- (b1) -- (c1);
    \draw[thick] (a2)-- (b2) -- (c2);
    \draw[thick] (a3)-- (b3) -- (c3);
    
	\end{tikzpicture}}
	\caption{The planar graph $G$ such that $\mathcal{R}_{\Gamma(G)+1}$ is not connected.}
	\label{fig:planar}
\end{figure}
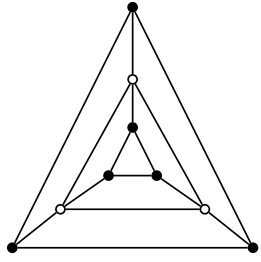

It is easily seen that $\Gamma(G)=3$. Moreover, if we consider the dominating set in white, in order to remove a vertex, we must add the two black vertices it is adjacent to, thus reaching a dominating set of size $\Gamma(G)+2$. We leave the question whether $\mathcal{R}_k(G)$ is connected if $G$ is planar and $k \ge \Gamma(G)+2$ as an open problem (see Conjecture~\ref{conjecture-planar}).

\section{Bounded treewidth graphs}\label{sec:twUB}

\begin{theorem}\label{them:twUB}
Let $G=(V,E)$ be a graph. If $k = \Gamma(G)+ tw(G) + 1$, then $\mathcal{R}_k(G)$ is connected. Moreover, the diameter of $\mathcal{R}_k(G)$ is at most $4(n+1) \cdot (tw(G)+1)$.
\end{theorem}

\begin{proof}
Let $(X,T)$ be a tree decomposition of $G$ such that the maximum size of a bag of $X$ is $tw(G)+1$. Let $b=|X|$. We root the tree $T$ in an arbitrary bag, then set $X:=\{X_1,\ldots,X_b\}$, where for any $X_i,X_j$ such that $X_i$ is a child of $X_j$, we have $i<j$. In other words, $X_1,\ldots,X_b$ is an elimination ordering of the (rooted) tree $T$ where at each step we remove a leaf of the remaining tree. We say that a bag $X_i$ is a \emph{descendant} of $X_j$ if $X_j$ is on the unique path from the root to $X_i$ (in other words, $X_i$ belongs to the subtree rooted in $X_j$ in $T$). Note that, free to contract edges if a bag is included in another, we can assume $b\leq n$. We denote by $V_i$ the set of vertices that do not appear in the set of bags $\cup_{j=i+1}^b X_j$. We set $V_0:= \emptyset$.

Let $D_\source$ and $D_\target$ be two dominating sets. Free to first remove vertices from $D_\source$ and $D_\target$ if possible (which can be done in at most $2(tw(G)+1)$ operations in total), we can assume that $D_\source$ and $D_\target$ have size at most $\Gamma(G)$. Let $D$ be a minimum dominating set of $G$. Instead of proving directly that there exists a reconfiguration sequence from $D_\source$ to $D_\target$, we will prove that that there exists a reconfiguration sequence from $D_\source$ to $D$ and from $D_\target$ to $D$ of length at most $2n \cdot (tw(G)+1)$ each. Since the reverse of a reconfiguration sequence also is reconfiguration sequence, that will give the conclusion, that gives a reconfiguration sequence of the desired length. So the rest of the proof is devoted to prove the following:

\begin{lemma}\label{lem:bdtw}
Let $G=(V,E)$ be a graph and let $D_\source$ be a dominating set of $G$ of size at most $\Gamma(G)$ and $D$ be a minimum dominating set of $G$. If $k = \Gamma(G)+ tw(G) + 1$, then there is a reconfiguration sequence from $D_\source$ to $D$. Moreover, the length of this reconfiguration sequence is at most $2n \cdot (tw(G)+1)$.
\end{lemma}

In order to prove Lemma~\ref{lem:bdtw}, we prove that there exists a sequence $\langle D_1:=D_\source,D_2,\ldots,D_b \rangle$ of dominating sets such that, for every $j$, $D_j$ satisfies the following property $\mathcal{P}$:
\begin{enumerate}[(i)]
    \item $D_j$ is a dominating set of $G$ of size at most $\Gamma(G)$,
    \item For every $j>1$, there exists a transformation sequence of length at most $2(tw(G)+1)$ from $D_{j-1}$ to $D_{j}$ in $\mathcal{R}_k(G)$,
    \item $D_j \cap V_{j-1} \subseteq D$. In other words, the vertices of $D_j$ that only belong to bags in $X_1 \cup \ldots \cup X_{j-1}$ are also in $D$.
\end{enumerate}
So that will provide a reconfiguration sequence in $\mathcal{R}_k(G)$ from $D_\source$ to a dominating set $D_b$ sufficiently close to $D$ to ensure the existence of a transformation from $D_b$ to $D$ of length at most $2n \cdot (tw(G)+1)$. To prove the existence of the sequence, we use induction on $j$.

First note that since $D_\source$ is a dominating set of $G$ of size at most $\Gamma(G)$ and $V_0$ is empty, $D_\source$ satisfies property $\mathcal{P}$. 
Let us now show that if $D_j$ satisfies property $\mathcal{P}$, then there exists a set $D_{j+1}$ that satisfies property $\mathcal{P}$. A vertex $v$ is a \emph{left vertex} (for $X_j$) if $v$ only appears in bags that are descendant of $X_j$. Note that by definition, $X_j$ is a descendant of itself. Otherwise, we say that $v$ is a \emph{right vertex}. When no confusion is possible, we will omit the mention of $X_j$.

\begin{claim}\label{rk:LeftRight}
If a left vertex $u$ (for $X_j$) is adjacent to a right vertex $v$ (for $X_j$), then $v\in X_j$.
\end{claim}

\begin{proofclaim}
Since $u$ and $v$ are adjacent in $G$, there exists a bag $X_i$ which contains both $u$ and $v$. Note that since $u$ is a left vertex, $X_i$ is a descendant of $X_j$. Besides, since $v$ is a right vertex, there exists a bag $X_{i'}$ that contains $v$ and which is not a descendant of $X_i$. Since the set of bags that contain $v$ induces a connected tree, $v$ must belong to each bag on the unique path from $X_i$ to $X_{i'}$. In particular, $v \in X_j$.
\end{proofclaim}

To construct $D_{j+1}$, we define several subsets of vertices (see Figure \ref{fig:tw} for an illustration).
\begin{itemize}
    \item $A$ is the set of left vertices of $X_j\cap (D_j\setminus D)$. In other words, $A$ is the set of vertices of $X_j$ that are in $D_j$ but not in $D$. 
    \item $B$ is the set of right vertices of $X_j$. In other words, $B$ is the set of vertices of $X_j$ that also appear in a bag $X_{j'}$ with $j' > j$.
    \item $C$ is the set of left vertices of $D\setminus D_j$. In other words, $C$ is the set of vertices of $D$ at the left of $X_j$ that are missing in $D_j$. 
    \end{itemize}
We partition again $B$ into three parts:
\begin{itemize}
    \item $B_1$ is the set of vertices of $B\setminus D$ that are dominated by $C$
    \item $B_2=B\cap D_j$
    \item $B_3=B\setminus (B_1\cup B_2)$.
\end{itemize}

\begin{figure}[bt]
	\centering
	\scalebox{0.9}
	{
	\begin{tikzpicture}[scale=0.08]
    
    \draw (0,0) circle (7cm);
    \draw (0,-40) circle (7cm);
    \draw (20,-20) circle (7cm);
    \draw (50,-20) circle (7cm);
    \draw (100,-20) circle (7cm);
    \draw (13,-20) node[left]{$X_j$};
    \draw (100,-27) node[below]{$X_b$};
    \draw (75,-20) node{$\ldots$};
    \draw ({cos(360-45)*7}, {sin(360-45)*7}) -- ({20+cos(90+45)*7}, {-20+sin(90+45)*7});
    \draw ({cos(45)*7}, {-40+sin(45)*7}) -- ({20+cos(180+45)*7}, {-20+sin(180+45)*7});
    \draw (27,-20) -- (43,-20);
    
    \draw[very thick,color=palette_bleuf] (2,2) -- (22,-16);
    \draw[very thick,color=palette_bleuf] (22,-16) -- (48,-16);
    \draw[very thick,color=palette_bleuf] (24,-18) -- (60,-18);
    
    \draw[very thick,dotted,color=palette_rouge] (2,-41) -- (18,-24);
    \filldraw[color=palette_rouge](21,-24) circle (10pt);
    
    \draw[thick,dashed, color=palette_vert] (-2,-2) -- (18,-20);
    \draw[thick,dashed, color=palette_vert] (-4,2) -- (-16,2);
    
    \draw[thick,dashed, color=palette_vert] (-60,-5) -- (-45,-5);
    \draw[thick,dotted, color=palette_rouge] (-60,-7) -- (-45,-7);
    
    \draw[thick,dashed, color=palette_vert] (-70,-32) -- (-52,-32);
    \draw[thick,dotted, color=palette_rouge] (-70,-34) -- (-52,-34);
    
    \draw[dashed, color=palette_vert] (-90,-21) -- (-62,-21);
    \draw[dotted, color=palette_rouge] (-90,-23) -- (-62,-23);
    
    \draw[color=palette_rouge] (22,-20) to[out=290,in=70] (22,-43);
    \draw[color=palette_rouge] (24,-31.5) node[right]{$A$};
    \draw[color=palette_vert] (-18,4) to[out=250,in=110] (-18,-22);
    \draw[color=palette_vert] (-22,-9) node[left]{$C$};
    \draw[color=palette_bleuf] (62,-20) to[out=70,in=290] (62,-14);
    \draw[color=palette_bleuf] (64,-17.5) node[right]{$B$};
    
	\end{tikzpicture}}
    \caption{The tree decomposition of $G$, and the sets $A$, $B$ and $C$. The circles represent the bags of the tree decomposition. The vertices are represented by lines, or dots, that go along the bags they belong to. The thick full lines represent the vertices of $B$, the dashed lines represent the vertices of $D$, and the dotted lines represent de vertices of $D_j$. By induction hypothesis, the left vertices of $D_j$ that do not belong to $X_j$ belong to $D$. }
	\label{fig:tw}
\end{figure}
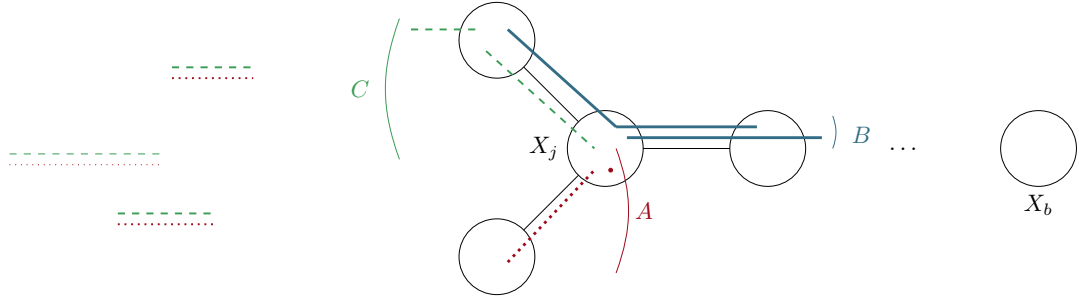

We set $D_j'=(D_j\setminus A)\cup C \cup B_3$. Let us first prove that $D_j'$ is a dominating set of $G$.

\begin{claim}\label{cl:Dr'}
The set $D_j'$ is a dominating set of $G$.
\end{claim}

\begin{proofclaim}
Since $D_j$ is a dominating set of $G$ and $D_j \setminus A \subseteq D_j'$, the only vertices that can be undominated in $D_j'$ are the ones dominated only by vertices of $A$ in $D_j$. Let $N_r(A)$ (resp. $N_l(A)$)) be the right vertices (resp. left vertices) that are only dominated by $A$ in $D_j$. Note that $N_l(A)$ might contain vertices of $A$, while $N_r(A)$ does not, since by definition the vertices of $A$ are left vertices. Let us show that all the vertices in $N_r(A) \cup N_l(A)$ are dominated by $D_j'$.

We start with $N_r(A)$. Since the vertices of $A$ are left vertices and the vertices of $N_r(A)$ are right vertices, by Claim~\ref{rk:LeftRight}, we have $N_r(A)\subseteq X_j$. Since the vertices in $N_r(A)$ are right vertices, we have $N_r(A)\subseteq B$. Moreover, since every vertex of $N_r(A)$ is only dominated by $A$ in $D_j$ but does not belong to $A$, it is not in $D_j$ and thus not in $B_2$. Thus, the vertices of $N_r(A)$ either belong to $B_1$ (and are by definition dominated by $C$), or they belong to $B_3$. Therefore, $N_r(A)$ is dominated by $C \cup B_3$ and thus by $D_j'$.

Let us now focus on $N_l(A)$. In $D$, $N_l(A)$ is dominated by vertices that we partition into two sets: the right vertices $Y$ and the left vertices $Z$. We show that both $Y$ and $Z$ are included in $D_j'$, which implies that $D_j'$ dominates $N_l(A)$. Since the vertices of $N_l(A)$ are left vertices and the vertices of $Y$ are right vertices, Lemma \ref{rk:LeftRight} gives $Y\subseteq X_j$. Thus, by definition, $Y\subseteq B$. Moreover, the vertices of $Y$ that belong to $D_j$ do not belong to $A$ as they are right vertices, and thus belong to $D_j\setminus A$, and the vertices of $Y$ that do not belong to $D_j$ belong by definition to $B\cap(D\setminus D_j)\subseteq B_3$. Thus, $Y\subseteq (D_j\setminus A)\cup B_3\subseteq D_j'$. Finally, the vertices of $Z$ either belong to $D_j$ and thus by definition to $D_j\cap D\subseteq D_j\setminus A$, or they do not belong to $D_j$ and by definition they thus belong to $C$. Therefore, $Z\subseteq (D_j\setminus A)\cup C\subseteq D_j'$. Therefore, $N_l(A)$ is dominated by $D_j'$, which concludes the proof of this claim.
\end{proofclaim}

Let us now prove the following:

\begin{claim}\label{cl:valueofk}
$|D_j\cup C \cup B_3|\leq \Gamma(G)+tw(G)+1.$
\end{claim}

\begin{proofclaim}
Let us first show that the set $D':=(D\setminus C)\cup A\cup B_1\cup B_2$ is a dominating set of $G$. We will then explain how to exploit this property to prove that $|D_j\cup C \cup B_3|\leq \Gamma(G)+tw(G)+1$.

Since $D$ is a dominating set, the only vertices that can be undominated in $(D\setminus C)\cup A\cup B_1\cup B_2$ are vertices that are only dominated by $C$ in $D$. Let $N_r(C)$ (resp. $N_l(C)$) be the subset of right (resp. left) vertices that are only dominated by $C$ in $D$. Note that $N_l(C)$ might contain vertices of $C$ and $N_r(C)$ does not, since the vertices of $C$ are left vertices. We prove that $N_r(C)$ and $N_l(C)$ are dominated by $D'$. 

We first prove that the vertices of $N_r(C)$ are dominated in $D'$.  Since $C$ only contains left vertices and $N_r(C)$ only contains right vertices, Claim~\ref{rk:LeftRight} ensures that $N_r(C)\subseteq X_j$. Thus, by definition of $B$, $N_r(C)\subseteq B$. Since the vertices of $N_r(C)$ are only dominated by $C$ in $D$, $N_r(C) \subseteq B_1$. Therefore $(D\setminus C)\cup A\cup B_1\cup B_2$ dominates $N_r(C)$.

Let us now prove that $N_l(C)$ is dominated in $D'$. Every vertex $v\in N_l(C)$ is dominated in $D_j$ by either a right vertex or a left vertex. Assume that $v$ is dominated in $D_j$ by a right vertex $w$. Since $v$ is a left vertex and $w$ a right vertex, Claim~\ref{rk:LeftRight} ensures that $w \in X_j$ and thus $w \in B$. Since $w \in D_j$, $w \in B_2\subseteq D'$. 
Assume now that $v$ is dominated in $D_j$ by a left vertex $u$. If $u$ belongs to $D$, it is in $D\cap D_j\subseteq D\setminus C \subseteq D'$. So we can assume that $u \notin D$. By induction hypothesis, $D_j$ satisfies (iii) and since $u \notin D$, the vertex $u$ necessarily belongs to $X_j$. So we finally have $u \in A$. Thus, $u \in  (D\setminus C)\cup A\subseteq D'$. So $N_l(C)$ is dominated in $D'$. And then $D'$ is a dominating set of $G$.

We can now show that $|D_j\cup C \cup B_3|\leq \Gamma(G)+tw(G)+1$. Since $D$ is a minimum dominating set of $G$ and $D'=(D\setminus C)\cup (A\cup B_1\cup B_2)$ also is a dominating set of $G$, we have $|C|\leq |A\cup B_1\cup B_2|$. Thus, $|C \cup B_3|\leq |A|+|B_1\cup B_2|+|B_3|$. But $A$, $B_1\cup B_2$ and $B_3$ are pairwise disjoint subsets of $X_j$. Thus, $|A|+|B_1\cup B_2|+|B_3|\leq |X_j| \leq tw(G)+1$, and $|C \cup B_3|\leq tw(G)+1$. Since, by induction hypothesis, $D_j$ has size at most $\Gamma(G)$, this gives $|D_j\cup C \cup B_3|\leq \Gamma(G)+tw(G)+1$.
\end{proofclaim}

We now have a reconfiguration sequence of size at most $tw(G)+1$ from $D_j$ to $D_j'$ by simply adding all the vertices of $C \cup B_3$ and then removing all the vertices of $A$. All along the sequence, the corresponding set is dominating. Indeed, it contains $D_j$ during the first part and $D_j'$ during the second one. One is dominating by assumption and the other is dominating by Claim~\ref{cl:Dr'}. By Claim~\ref{cl:valueofk}, this reconfiguration sequence exists in $\mathcal{R}_{\Gamma(G)+tw(G)+1}(G)$.

The dominating set $D_{j+1}$ will be any dominating set of size at most $\Gamma(G)$ obtained from $D_j'$ by removing vertices, i.e.\ any dominating set $D_{j+1}$ satisfying $D_{j+1} \subseteq D_j'$ and $|D_{j+1}|=\Gamma(G)$, which necessarily exist by definition of $\Gamma(G)$.
This can be done in at most $tw(G)+1$ deletions. Thus, there exist a sequence in $\mathcal{R}_{\Gamma(G)+tw(G)+1}(G)$ from $D_j$ to $D_{j+1}$ of length at most $2(tw(G)+1)$, and $D_{j+1}$ thus satisfies (i) and (ii). Let us now justify why $D_{j+1}$ satisfies (iii).

Since $D_{j+1}$ is a subset of $D_{j}'$, if (iii) holds for $D_j'$ it holds for $D_{j+1}$. We have $D_j'=(D_j\setminus A)\cup C \cup B_3$.
Since $C \subseteq D$, if a left vertex $v$ (for $X_j$) appears in $D_j'$ but not in $D$, it is either in $D_j \setminus A$ or in $B_3$. Since $B_3$ only contains right vertices, it must be in $D_j \setminus A$. Since $A$ contains the left vertices of $X_j \cap (D_j \setminus D)$, it means that $v$ should be in $V_{j-1}$. But, by induction hypothesis, the vertices of $D_j$ that belong to $V_{j-1}$ belong to $D$. So $v$ does not exists and $D_j'$ satisfies (iii). Thus, $D_{j+1}$ satisfies property $\mathcal{P}$, and by induction, there exists a set $D_{b}$ that satisfies property $\mathcal{P}$. Moreover, since for any $i$ such that $2\leq i\leq b$, there is a path of length at most $2(tw(G)+1)$ from $D_{i-1}$ to $D_{i}$ in $\mathcal{R}_k(G)$, there is transformation of length at most $2(b-1)\cdot(tw(G)+1)$ from $D_\source$ to $D_b$ in $\mathcal{R}_k(G)$. 
\medskip

To complete the construction of a path from $D_\source$ to $D$ in $\mathcal{R}_k(G)$, we show that there exists a transformation from $D_b$ to $D$ in $\mathcal{R}_k(G)$ of length at most $2(tw(G)+1)$.
Let $A'=D_b\setminus D$, and $C'=D\setminus D_b$. We have $D=(D_b\cup C')\setminus A'$. Let $S_1'$ be the reconfiguration sequence from $D_b$ to $D_b\cup C'$ which consists in adding one by one every vertex of $C'$. Since each of the sets of $S_1'$ contains $D_b$, they are all dominating sets of $G$. Note that $S_1'$ has length $|C'|$. Let $S_2'$ be the reconfiguration sequence from $D_b\cup C'$ to $D$ which consists in removing one by one each vertex of $A'$. Since each of the sets of $S_2'$ contains $D$, they all are dominating sets. Note that $S_2'$ has length $|A'|$. Thus, applying $S_1'$ then $S_2'$ gives a reconfiguration sequence from $D_b$ to $D$ of length $|C'|+|A'|$. Moreover, the maximum size of a dominating set reached in this sequence is $|D_b\cup C'|$. Let us show that $|D_b\cup C'|\leq \Gamma(G)+tw(G)+1$. We have $D_b=(D\setminus C')\cup A'$. Thus, since $D$ is a minimum dominating set, $|C'| \le |A'|$. Since $D_b$ satisfies (iii), every vertex of $D_b$ that does not belong to $X_b$ also belongs to $D$. Thus, $A'\subseteq X_b$, and $|A'|\leq tw(G)+1$, which gives $|C'|\leq tw(G)+1$, as well as $|C'|+|A'|\le 2(tw(G)+1)$. Since $D_b$ is a minimal dominating set of $G$, we have therefore $|D_b\cup C'|\leq \Gamma(G)+tw(G)+1$. Thus, there is a path of length at most $2(tw(G)+1)$ from $D_b$ to $D$ in $\mathcal{R}_k(G)$ which completes the transformation of length at most $2b\cdot(tw(G)+1)$ from $D_\source$ to $D$ in $\mathcal{R}_k(G)$. Since $b\leq n$, the conclusion follows.
\end{proof}

\bibliographystyle{abbrv}
\bibliography{bibliography}

\end{document}